\documentclass[aps,prl,reprint,superscriptaddress,longbibliography]{revtex4-2}
\usepackage[T1]{fontenc}
\usepackage{amsmath,amssymb,amsthm,mathtools,graphicx,tikz}
\usepackage{microtype}
\usepackage{xcolor}
\usepackage[colorlinks=true,allcolors=blue!50!black]{hyperref}
\hypersetup{pdftitle={Quantum Behaviors Are Not Semialgebraic},
 pdfsubject={Exact geometry of quantum correlations in a fixed Bell scenario},
 pdfkeywords={quantum correlations, Bell nonlocality, semialgebraic geometry, semidefinite programming}}
\newcommand{\Cq}{\mathcal C_q}
\newcommand{\Cqa}{\mathcal C_{qa}}
\newcommand{\Cqc}{\mathcal C_{qc}}
\newcommand{\Ct}{\mathcal C_t}
\newcommand{\id}{I}
\DeclareMathOperator{\Tr}{Tr}
\DeclareMathOperator{\diag}{diag}
\newcommand{\ip}[2]{\langle #1,#2\rangle}
\newcommand{\norm}[1]{\lVert #1\rVert}

\newtheorem{theorem}{Theorem}
\newtheorem{lemma}[theorem]{Lemma}
\newtheorem{proposition}[theorem]{Proposition}
\newtheorem{corollary}[theorem]{Corollary}
\theoremstyle{remark}\newtheorem{remark}[theorem]{Remark}
\newcommand{\HH}{\mathcal H}
\DeclareMathOperator{\spec}{spec}
\DeclareMathOperator{\spanop}{span}
\definecolor{sliceblue}{HTML}{155E86}
\definecolor{slicefill}{HTML}{A7D4E5}

\begin{document}
\title{Quantum Behaviors Are Not Semialgebraic}

\author{Minbo Gao}%
\email{gaomb@ios.ac.cn}%
\affiliation{Key Laboratory of System Software (Chinese Academy of Sciences),
  Institute of Software, Chinese Academy of Sciences, Beijing, P.R. China}%
\affiliation{University of Chinese Academy of Sciences, Beijing, P.R. China}

\author{Zhengfeng Ji}%
\email{jizhengfeng@tsinghua.edu.cn}%
\affiliation{ Department of Computer Science and Technology, Tsinghua
  University, Beijing, P.R. China }%

\author{Chenghua Liu}%
\email{liuch@ios.ac.cn}%
\affiliation{Key Laboratory of System Software (Chinese Academy of Sciences),
  Institute of Software, Chinese Academy of Sciences, Beijing, P.R. China}%
\affiliation{University of Chinese Academy of Sciences, Beijing, P.R. China}

\begin{abstract}
The conditional probabilities achievable by local measurements on a
shared quantum state in a Bell scenario form the set of quantum
behaviors, whose structure was studied by Tsirelson. In 1993, Tsirelson
asked whether this set is semialgebraic, that is, describable by finite
Boolean combinations of polynomial equations and inequalities. This
question has remained open. We answer it in the negative: with four
binary measurements per party, the set of finite-dimensional quantum
behaviors, its closure, and the commuting-operator set are all
nonsemialgebraic. More strongly, none admits a finite real-analytic
description even locally near a particular classical behavior. These
results rule out exact finite semidefinite representations and show
that no finite level of the Navascu\'es--Pironio--Ac\'in hierarchy
characterizes these sets exactly.
\end{abstract}
\maketitle

Can finitely many polynomial equations and inequalities characterize
all quantum behaviors in a Bell experiment? Posed by Tsirelson in
1993~\cite{Tsirelson}, this longstanding question asks for a finite
description of the conditional probability tables achievable by local
measurements on a shared quantum state---the \emph{quantum behaviors}.
Characterizing this set is central to Bell nonlocality and
device-independent information processing, where observable statistics
reveal the capabilities of uncharacterized devices~\cite{Bell,Brunner}.
In a fixed Bell scenario, the local and nonsignaling sets are polytopes,
each defined by finitely many linear inequalities. The quantum set lies
between them but has a richer geometry~\cite{Goh}, making it natural to
ask whether finitely many nonlinear conditions can describe it exactly.

Finite characterizations do exist in important restricted settings.
Binary correlators, with single-party marginals omitted, admit a finite
positive-semidefinite Gram matrix characterization~\cite{Tsirelson87}.
For full behaviors with two binary measurements per party, a reduction
to qubits followed by finite convexification yields a semialgebraic
description~\cite{Masanes}. The minimal correlator body is known
to be semialgebraic but not \emph{basic} semialgebraic~\cite{Le}:
finite unions of polynomial regions can succeed where a single
conjunction of inequalities fails. At any fixed Hilbert-space dimension,
the conditions on states and measurements, together with the Born rule,
can be expressed as polynomial constraints on finitely many matrix
entries. Eliminating these entries preserves
semialgebraicity by the Tarski--Seidenberg theorem~\cite{BCR}.
The question is whether such a finite description survives when the
internal dimension is unbounded but the Bell scenario remains fixed.

The original question also allows analytic inequalities
(Ref.~\cite{Tsirelson}, Problem~2.10). For the semialgebraic question, a negative answer beyond the smallest
scenario was later anticipated by Ozawa~\cite{Ozawa2013}.
Here we consider the full semialgebraic class: finite Boolean
combinations of polynomial conditions with arbitrary real coefficients.
Allowing finitely many auxiliary real variables does not enlarge this
class~\cite{BCR}. For the local analytic question, the defining
functions are required to be real analytic on an open neighborhood
of the behavior being described.

In this Letter, we give an elementary proof that the set of quantum
behaviors is not semialgebraic in the Bell scenario with four binary
measurements per party. We construct an explicit quadratic curve that
intersects the
finite-dimensional quantum behaviors, its closure, and the commuting-operator
quantum behaviors at the same infinite discrete sequence of points.
Interestingly, these points
converge to a point with a simple classical local model.
Every finite semialgebraic outer approximation must instead contain
an entire interval of the curve near this limit, revealing a specific
limitation of
every finite level of the Navascu\'es--Pironio--Ac\'in (NPA)
hierarchy~\cite{NPA2007,NPA}.

Nonclosure~\cite{Slofstra,DPP}, separation of the closed tensor-product
and commuting-operator models~\cite{MIPRE}, and fixed-scenario
membership undecidability~\cite{FMS} establish other limitations on
the characterization of quantum behaviors.
A nonclosed set can still be semialgebraic; undecidability alone does
not exclude descriptions with arbitrary, possibly noncomputable, real
coefficients. Our obstruction is directly geometric and persists under
closure. It applies to finite semialgebraic relaxations regardless of
how they are constructed, including those obtained within the
noncommutative polynomial-optimization framework~\cite{PNA}.

The construction builds on the known structure of synchronous correlations,
for which equal inputs always give equal outputs~\cite{PSSTW}.
Kruglyak, Rabanovich, and Samoilenko classified scalar sums of four
projections~\cite{KRS}. The correlation formulas used here and the
finite-dimensional argument that recovers a scalar projection sum from
their moments already appear in the self-testing work of Man\v{c}inska,
Prakash, and Schafhauser~\cite{MPS}. We use this structure to determine
the full intersection of the curve with each of the three sets
and to derive the local analytic obstruction and its consequences for
finite exact descriptions.

\emph{Models and the quadratic curve.---}
Inputs are $i,j\in\{1,2,3,4\}$ and outputs are $a,b\in\{0,1\}$.
A behavior is the vector $p=(p(a,b\mid i,j))\in\mathbb R^{64}$.
The finite-dimensional quantum set $\Cq$ consists of
\begin{equation}
 p(a,b\mid i,j)=\Tr[\rho(E_i^a\otimes G_j^b)],
 \label{eq:model}
\end{equation}
with arbitrary finite local dimensions, a density operator $\rho$, and
local positive-operator-valued measurements (POVMs). Write
$\Cqa=\overline{\Cq}$ for its Euclidean closure. The set $\Cqc$ consists
of behaviors $\ip{\psi}{A_i^aB_j^b\psi}$ on an arbitrary Hilbert space,
where $\psi$ is a unit vector and $A_i=A_i^1$, $B_j=B_j^1$ are orthogonal
projections satisfying $[A_i,B_j]=0$; outcome-$0$ projections are their
complements. Superscripts label outcomes. Standard dilation and moment-limit
arguments give $\Cq\subseteq\Cqa\subseteq\Cqc$~\cite{NPA}
(see the Supplemental Material~\cite{SM} for details).
Both inclusions are known to be strict in general: the first by
Slofstra's nonclosure theorem~\cite{Slofstra}, and the second by the
$\mathrm{MIP}^*=\mathrm{RE}$ theorem~\cite{MIPRE}.

For a real parameter $\alpha$, set
\begin{equation}
 u=\frac\alpha4,\qquad v=\frac{\alpha(\alpha-1)}{12}.
 \label{eq:uv}
\end{equation}
In output order $(11,10,01,00)$, define the quadratic curve
\begin{align}
 p_\alpha(\cdot\mid i,i)&=(u,0,0,1-u),\nonumber\\
 p_\alpha(\cdot\mid i,j)&=(v,u-v,u-v,1-2u+v)\quad(i\ne j).
 \label{eq:curve}
\end{align}
For $1<\alpha<2$ these are normalized, nonnegative, nonsignaling
behaviors, with outcome-$1$ marginals $u$. They are synchronous:
equal inputs always give equal outputs.

\begin{theorem}[Exact intersection and local obstruction]\label{thm:main}
Write $\alpha_m=2-2/m$ for $m\ge3$. For each
$t\in\{q,qa,qc\}$, in the four-input, binary-output scenario,
\begin{equation}
 \{\alpha\in(1,2):p_\alpha\in\Ct\}
 =\left\{2-\frac2m:m=3,4,\ldots\right\}.
 \label{eq:exact}
\end{equation}
Consequently, $\Cq$, $\Cqa$, and $\Cqc$ are not semialgebraic over
$\mathbb R$. Moreover, none is semianalytic at the classical behavior $p_2$.
\end{theorem}

The exact intersection is illustrated in Fig.~\ref{fig:slice}.
We prove necessity in the largest model and construct every allowed
point in the smallest one. Full operator-theoretic details are provided
in the Supplemental Material~\cite{SM}.

\begin{figure}[t]
 \centering
 \resizebox{\columnwidth}{!}{%
 \begin{tikzpicture}[x=10cm,y=1cm,font=\footnotesize]
  \def\betaval{1.69}
  \fill[slicefill] (\betaval,-0.18) rectangle (2,0.18);
  \foreach \y in {0,1.65} {
   \draw[gray,line width=0.35pt] (1.25,\y) -- (2.025,\y);
   \foreach \m in {3,...,200} {
    \fill[sliceblue] ({2-2/\m},\y) circle[radius=1.0pt];
   }
   \fill[black] (1.995,\y-0.05) rectangle (2.005,\y+0.05);
   \node[anchor=west,inner sep=1pt] at (2.025,\y) {$\alpha$};
  }
  \draw[sliceblue,line width=0.6pt] (\betaval,-0.23) -- (\betaval,0.23);
  \node[anchor=west,inner sep=0pt] at (1.25,2.25)
       {Exact: $\Cq,\ \Cqa,\ \Cqc$};
  \node[anchor=east,inner sep=0pt] at (2,2.25) {$p_2$ local};
  \node[anchor=west,inner sep=0pt] at (1.25,0.60)
       {Any finite NPA level $\mathcal Q_\ell$};
  \foreach \value/\ticklabel in {{4/3}/{4/3},{3/2}/{3/2},{7/4}/{7/4},2/2} {
   \draw[gray,line width=0.35pt] ({\value},1.55) -- ({\value},1.60);
   \node[anchor=north,inner sep=1pt] at ({\value},1.42) {$\ticklabel$};
  }
  \node[anchor=north,inner sep=1pt] at (\betaval,-0.30) {$\beta_\ell$};
  \node[anchor=north,inner sep=1pt] at ({(\betaval+2)/2},-0.66)
       {guaranteed interval};
 \end{tikzpicture}%
 }
 \caption{Exact and relaxed realizability along the quadratic curve
 $p_\alpha$. Top: all three models share the parameters
 $\alpha_m=2-2/m$, accumulating at the classical endpoint $p_2$
 (square). Bottom: every finite NPA level contains a tail interval
 $(\beta_\ell,2)$ (shaded). All other parameters in that interval lie
 outside even $\Cqc$. Only finitely many exact points are drawn;
 $\beta_\ell$ is schematic. Figure drawn with Codex-assisted TikZ code.}
 \label{fig:slice}
\end{figure}
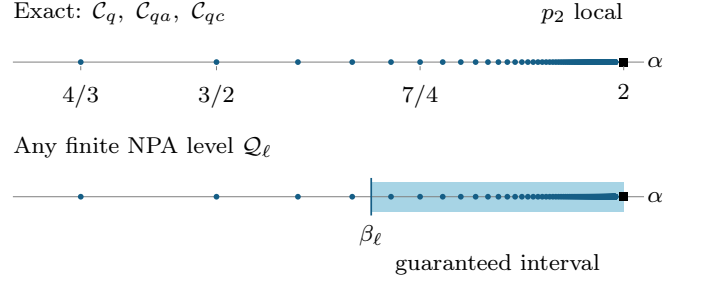

\emph{From probabilities to projections.---}
For $p_\alpha\in\Cqc$, take a commuting representation
$A_i,B_j,\psi$. Following the moment argument of Ref.~\cite{MPS},
the synchrony relation~\cite{PSSTW} gives
\begin{equation}
 \norm{(A_i-B_i)\psi}^2
 =p_\alpha(1,0\mid i,i)+p_\alpha(0,1\mid i,i)=0.
 \label{eq:sync}
\end{equation}
Thus $A_i\psi=B_i\psi$. Let $S=\sum_iA_i$. Replacing $A_j\psi$
by $B_j\psi$ in each second moment gives
\begin{align}
 \norm{(S-\alpha\id)\psi}^2
 &=\sum_{i,j}p_\alpha(1,1\mid i,j)+\alpha^2\nonumber\\
 &\quad-2\alpha\sum_i p_\alpha(1,1\mid i,i)\nonumber\\
 &=\alpha^2-2\alpha^2+\alpha^2=0.
 \label{eq:variance}
\end{align}
This vector identity becomes an operator identity on the nonzero closed
subspace
\begin{equation}
 K=\overline{\operatorname{span}}\{w(B_1,\ldots,B_4)\psi:
                                  w\text{ a word}\},
 \label{eq:cyclic}
\end{equation}
including the empty word. The relation
\[
 A_iw(B)\psi=w(B)B_i\psi\in K
\]
shows that the restrictions $P_i=A_i|_K$
are orthogonal projections. Since $S$ commutes with every Bob word,
Eq.~\eqref{eq:variance} implies that $S-\alpha\id$ vanishes on the
dense span in Eq.~\eqref{eq:cyclic}. By continuity, it vanishes on all
of $K$, giving
\begin{equation}
 P_1+P_2+P_3+P_4=\alpha\id_K.
 \label{eq:scalar}
\end{equation}
No faithful state or finite-dimensional representation is required.
See Lemma~\ref{supp:lem:scalar} in the Supplemental Material~\cite{SM}
for details.

The scalar values allowed by Eq.~\eqref{eq:scalar} are given by the
following known classification, stated here in our notation.

\emph{KRS classification (Ref.~\cite{KRS}, Proposition~3(a)).}
There exist four orthogonal projections $P_1,\ldots,P_4$ on a nonzero
complex Hilbert space $K$ satisfying $\sum_{i=1}^4P_i=\alpha\id_K$
if and only if
\begin{equation}
 \alpha\in\Sigma_4:=\{0,1,2,3,4\}\cup
 \left\{2\pm\frac{2}{m}:m=3,4,\ldots\right\}.
 \label{eq:krs-classification}
\end{equation}
Here each $P_i$ is selfadjoint and idempotent; the projections need not
be mutually orthogonal.

In particular, for $1<\alpha<2$, Eq.~\eqref{eq:scalar} can hold only when
$\alpha=2-2/m$ for an integer $m\ge3$.
We include a self-contained proof of this fact in the
Supplemental Material~\cite{SM}.
Very briefly, if $A=P+Q$ and $D=P-Q$, then $D^2=A(2\id-A)$ and $AD=D(2\id-A)$;
hence spectral points of $A$ in $(0,2)$ occur in pairs $\lambda,2-\lambda$.
Applying this reflection to $P_1+P_2$ and $P_3+P_4=\alpha\id-P_1-P_2$ gives a
descent by $4-2\alpha$.
Positivity of the spectrum forces that descent to terminate, quantizing $\alpha$ to the stated
sequence.
This proves the necessary inclusion in Eq.~\eqref{eq:exact}.

\emph{Realizing every allowed parameter.---}
Let $m\ge3$ and $\alpha=2-2/m$. On $\mathbb C^m$, put
\begin{equation}
 A=\diag\!\left(\frac{2j}{m}:0\le j\le m-1\right),
 \qquad B=\alpha\id-A.
 \label{eq:diagonal}
\end{equation}
The diagonal of $B$ is the reverse of that of $A$.
Apart from a single $0$ and, for even $m$, a single $1$, the entries
pair as $\lambda,2-\lambda$.
Each such pair is a sum of the two projections
\begin{equation}
 P_\lambda^\pm=\frac12
 \begin{pmatrix}
 \lambda&\pm\sqrt{\lambda(2-\lambda)}\\
 \pm\sqrt{\lambda(2-\lambda)}&2-\lambda
 \end{pmatrix}.
 \label{eq:blocks}
\end{equation}
The single entries decompose as $0=0+0$ and $1=1+0$.
Thus $A=P_1+P_2$ and $B=P_3+P_4$ for four real projections
with sum $\alpha\id$.

Symmetrize the labels by taking
$Q_i=\bigoplus_{\pi\in S_4}P_{\pi(i)}$ on dimension $d=24m$.
For the normalized trace $\tau=\Tr/d$, permutation symmetry and
$\sum_iQ_i=\alpha\id$ imply
\begin{equation}
 \tau(Q_i)=\frac\alpha4,\qquad
 \tau(Q_iQ_j)=\frac{\alpha(\alpha-1)}{12}\quad(i\ne j).
 \label{eq:traces}
\end{equation}
The second identity follows by expanding $\tau[(\sum_iQ_i)^2]$.
On the maximally entangled state
$\Phi_d=d^{-1/2}\sum_{r=1}^d e_r\otimes e_r$, use $Q_i$ for
Alice's outcome $1$ and $Q_j^{\mathsf T}$ for Bob's. The identity
\begin{equation}
 \ip{\Phi_d}{X\otimes Y^{\mathsf T}\Phi_d}=\tau(XY)
 \label{eq:ME}
\end{equation}
reproduces Eq.~\eqref{eq:curve}, including the marginals and the
zero mixed-output probabilities at equal inputs. Hence every
$p_{2-2/m}$ belongs to $\Cq$. Together with
$\Cq\subseteq\Cqa\subseteq\Cqc$, this proves Eq.~\eqref{eq:exact}.

If any of the three correlation sets were semialgebraic, substituting
the polynomial curve into its finite Boolean description would make
the parameter set in Eq.~\eqref{eq:exact} semialgebraic in one variable.
A finite collection of nonzero univariate polynomials has only finitely
many roots, and its signs are constant between consecutive roots.
Every semialgebraic subset of the line is therefore a finite union
of intervals and points. It cannot be infinite and countable.
This contradiction proves the nonsemialgebraicity assertion of Theorem~\ref{thm:main},
without any restriction on the polynomial coefficients.

\emph{An analytic obstruction at a classical behavior.---}
The endpoint $p_2$ is classical: the parties share a uniformly random
two-element subset of $\{1,2,3,4\}$ and output $1$ exactly when their
input belongs to it. The resulting marginals are $1/2$ and the
distinct-input $(1,1)$ probabilities are $1/6$, as in Eq.~\eqref{eq:curve}.
Nevertheless, none of $\Cq$, $\Cqa$, or $\Cqc$ is semianalytic at $p_2$.
Indeed, suppose membership in a neighborhood of $p_2$ were a finite
Boolean combination of signs of real analytic functions. Composing
them with $p_\alpha$ gives finitely many functions analytic near
$\alpha=2$. Each is either identically zero near $2$ or has no zeros
in a sufficiently small punctured neighborhood. Their signs, and
therefore membership, are constant on some interval $(2-\epsilon,2)$.
Equation~\eqref{eq:exact} contradicts this. This statement requires
analyticity at the limit point; it makes no claim about functions
singular there or arbitrary projections of analytic sets.

\emph{What the obstruction means physically.---}
The accumulation point is classical, but every realizable $p_\alpha$
with $1<\alpha<2$ is Bell nonlocal. In a synchronous local model the
outputs are shared predetermined bits $z_1,\ldots,z_4$. The same zero
variance as in Eq.~\eqref{eq:variance} would require
$\sum_i z_i=\alpha$ with probability one, impossible for noninteger
$\alpha$. Nonlocal quantum points therefore approach a behavior with
an elementary local explanation.

The obstruction is invisible if only the binary correlators are kept.
Writing $C_{ij}=\sum_{a,b}(-1)^{a+b}p_\alpha(a,b\mid i,j)$ gives
\begin{equation}
 C_{ii}=1,\qquad C_{ij}=c_\alpha:=\frac{(\alpha-2)^2-1}{3}\quad(i\ne j).
 \label{eq:correlators}
\end{equation}
For every $1<\alpha<2$, these correlators themselves admit a classical
model. Mix uniformly random balanced four-sign strings with uniformly
random unrestricted four-sign strings, using weights $-3c_\alpha$ and
$1+3c_\alpha$; both parties output the sign at their input. This
reproduces Eq.~\eqref{eq:correlators}, with zero single-party means.
The full table additionally fixes those means to $1-\alpha/2$.
Their compatibility with the correlators is what restricts $\alpha$
to a discrete sequence. Thus the finite characterizations available
for correlators alone~\cite{Tsirelson87,Le} cannot characterize the
full quantum behavior.

The same projection identity makes the dimension cost explicit.
Let $D_m^{\mathrm{sr}}$ be the smallest bound on both local quantum
dimensions for an exact realization of $p_{\alpha_m}$, with shared
classical randomness free. Then
\begin{equation}
 \frac{m}{\gcd(m,2)}\le D_m^{\mathrm{sr}}\le m.
 \label{eq:dimension}
\end{equation}
For the lower bound, synchrony turns the POVMs into projections on
each pure component's Schmidt supports. Zero variance forces their
sum to be $\alpha_m I$; taking the trace makes the Schmidt rank a
multiple of the denominator of $\alpha_m$. The upper bound implements
the label symmetrization using shared randomness on dimension $m$.
Thus the minimum dimension for exact realization grows as
$\Theta((2-\alpha_m)^{-1})$, although the limiting behavior requires
only shared randomness. This elementary rank argument places
established dimension-sensitive phenomena~\cite{NV,MPS} at the
classical accumulation point. For odd $m$, the same behaviors already
robustly self-test an $m$-dimensional maximally entangled state~\cite{MPS}.
The proof for all $m$, including mixed states and arbitrary shared
randomness, is given in the Supplemental Material~\cite{SM}.

There is also a sharp comparison with smaller measurement scenarios.
Russell's results give a uniformly bounded-dimensional realization of
the synchronous binary sector with three inputs and identify its
finite-dimensional and commuting models~\cite{RussellGeom,RussellConnes}.
That sector is semialgebraic. Theorem~\ref{thm:main} therefore places
the onset of nonsemialgebraicity at exactly four inputs
\emph{within the synchronous binary sector}
(see the Supplemental Material~\cite{SM}).

\emph{Consequences for exact descriptions and relaxations.---}
First, none of the three models admits a finite semialgebraic lift:
there are no finite $r$ and semialgebraic set $Z\subset\mathbb R^{64+r}$
such that
\begin{equation}
 \Ct=\{p:\text{there exists }z\in\mathbb R^r\text{ with }(p,z)\in Z\}.
 \label{eq:lift}
\end{equation}
The Tarski--Seidenberg theorem~\cite{BCR} would make this projection semialgebraic.
In particular, no fixed finite collection of linear matrix inequalities
with auxiliary variables represents $\Ct$ exactly: matrix positivity
is equivalent to finitely many principal-minor inequalities. Allowing
arbitrary real matrix entries does not remove the obstruction.

Second, the curve exposes how every finite semialgebraic outer approximation
must fail, independently of its construction.

\begin{corollary}\label{cor:tail}
If $R\subset\mathbb R^{64}$ is semialgebraic and $\Cq\subseteq R$,
there is $\beta\in(1,2)$ such that $p_\alpha\in R$ for all
$\alpha\in(\beta,2)$.
\end{corollary}

To see this, the inverse image
$\{\alpha\in(1,2):p_\alpha\in R\}$ is semialgebraic and contains
the sequence $2-2/m\to2$. Its finite interval-and-point decomposition
must therefore contain a whole interval ending at $2$.
Theorem~\ref{thm:main} shows that every parameter in this interval outside the
discrete sequence produces a behavior outside even $\Cqc$.

Each fixed NPA level $\mathcal Q_\ell$ is the projection of a finite
semidefinite feasibility problem and contains $\Cqc$~\cite{NPA}.
Corollary~\ref{cor:tail} gives
\begin{equation}
 \exists\beta_\ell<2\quad
 p_\alpha\in\mathcal Q_\ell\quad
 \text{for every }\alpha\in(\beta_\ell,2).
 \label{eq:tail}
\end{equation}
No finite NPA level can therefore equal either $\Cqa$ or $\Cqc$
in this scenario. Completeness guarantees eventual exclusion of every
fixed behavior outside $\Cqc$, but there is no uniform finite level
that resolves all the gaps near $\alpha=2$.

The same conclusion applies to other finite semialgebraic outer
approximations. For example, the almost-quantum set has a finite
semidefinite description~\cite{AlmostQuantum} and must therefore
contain an interval of the same kind. This constrains proposed physical
principles whenever their characterization reduces to finitely many
polynomial conditions and auxiliary variables; principles involving
unlimited composition or limiting procedures are outside this scope.

On this curve, exact membership for rational $1<\alpha<2$ is decided
simply by checking whether $2/(2-\alpha)$ is an integer at least three.
Thus an elementary membership rule can coexist with an obstruction to
every finite polynomial description. The obstruction transfers to all
larger Bell scenarios containing the four-input, binary-output case
by duplicating measurements and adding unused outcomes
(see the Supplemental Material~\cite{SM}).

These statements concern exact membership. Synchrony places the curve
on the nonsignaling boundary, and the spacing
$\alpha_{m+1}-\alpha_m=2/[m(m+1)]$ vanishes at the accumulation point.
Quantifying distances of forbidden behaviors from the quantum set,
the NPA order needed to exclude them, and the effect of imperfect
synchrony would determine how this structure appears at finite
precision. The exact result already establishes that a classical
limiting behavior need not have a quantum neighborhood describable
by finitely many analytic conditions.

\paragraph*{Data availability.}
This work contains analytical results. All constructions and proofs
supporting the conclusions are provided in the article and Supplemental
Material; no empirical data were generated or analyzed.

\begin{acknowledgments}
OpenAI Codex was used to assist with literature searches, proof
development and mathematical checks, manuscript preparation, and
TikZ code for the schematic figure.
\end{acknowledgments}

\onecolumngrid
\clearpage
\setcounter{secnumdepth}{1}
\setcounter{section}{0}
\setcounter{equation}{0}
\setcounter{theorem}{0}
\renewcommand{\thesection}{S\arabic{section}}
\renewcommand{\theequation}{S\arabic{equation}}
\renewcommand{\thetheorem}{S\arabic{theorem}}
\renewcommand{\theHsection}{supp.\arabic{section}}
\renewcommand{\theHequation}{supp.\arabic{equation}}
\renewcommand{\theHtheorem}{supp.\arabic{theorem}}
\allowdisplaybreaks[1]
\setlength{\emergencystretch}{2em}
\phantomsection
\label{supplement-start}
\begin{center}
 {\large\bfseries Supplemental Material}\par\medskip
 {\bfseries Quantum Behaviors Are Not Semialgebraic}
\end{center}
\medskip
This supplement provides complete proofs of the spectral restriction,
the construction at every allowed parameter, the passage to limiting
correlations, and the geometric consequences stated in the Letter.
The scalar classification is due to Ref.~\cite{KRS}; the correlation
formulas and finite-dimensional moment-to-scalar argument have precedents
in Ref.~\cite{MPS}. The derivations here make the geometric argument
self-contained. Proof development and checking were assisted by OpenAI
Codex, as disclosed in the Acknowledgments; the complete arguments are
given below.

\section{Models and notation}
The inputs are $i,j\in\{1,2,3,4\}$ and the outputs $a,b\in\{0,1\}$.
Write $\Cq=\Cq(4,4;2,2)$ for all behaviors
\begin{equation}\label{supp:eq:tensor-model}
 p(a,b\mid i,j)=\Tr[\rho(E_i^a\otimes G_j^b)]
\end{equation}
on arbitrary finite-dimensional local complex Hilbert spaces. Here
$\rho$ is a density operator, $E_i^a,G_j^b\succeq0$, and
$\sum_aE_i^a=\sum_bG_j^b=\id$. Define
$\Cqa=\overline{\Cq}$ in the Euclidean space $\mathbb R^{64}$.
The set $\Cqc$ consists of behaviors
$p(a,b\mid i,j)=\ip{\psi}{A_i^aB_j^b\psi}$ on an arbitrary complex
Hilbert space, with a unit vector $\psi$ and orthogonal projections
$A_i,B_j$ commuting across parties. Set $A_i^1=A_i$, $A_i^0=\id-A_i$
and likewise for $B_j$. Superscripts label outcomes, not powers.
Lemma~\ref{supp:lem:limit} proves $\Cq\subseteq\Cqa\subseteq\Cqc$ without
assuming equality of these models.

Inner products are conjugate-linear in the first argument. An orthogonal
projection is a selfadjoint idempotent; distinct projections need not be
mutually orthogonal. Spectra are those of bounded operators.
All references to semialgebraicity allow arbitrary real coefficients and
arbitrary finite Boolean combinations of polynomial conditions.

\section{Scalar sums of four projections}
\label{supp:sec:projections}
\begin{lemma}[Spectral reflection]\label{supp:lem:reflection}
If $P,Q$ are orthogonal projections on a nonzero complex Hilbert space
and $\lambda\in\spec(P+Q)\cap(0,2)$, then
$2-\lambda\in\spec(P+Q)$.
\end{lemma}
\begin{proof}
Set $A=P+Q$ and $D=P-Q$. Direct multiplication gives
\begin{equation}\label{supp:eq:reflection-identities}
 D^2=A(2\id-A),\qquad AD=D(2\id-A).
\end{equation}
Every spectral point $\lambda$ of a bounded selfadjoint operator $A$
has unit approximate eigenvectors $v_n$ satisfying
$\norm{(A-\lambda\id)v_n}\to0$.
Indeed, if $\norm{(A-\lambda\id)v}\ge c\norm v$ for some $c>0$
and all $v$, then $A-\lambda\id$ has zero kernel and closed range.
Its range is dense because its orthogonal complement is
$\ker(A-\lambda\id)^*=\ker(A-\lambda\id)=\{0\}$.
Thus it is bijective with bounded inverse, a contradiction.

For these approximate eigenvectors, \eqref{supp:eq:reflection-identities} gives
\[
 \norm{Dv_n}^2=\ip{v_n}{A(2\id-A)v_n}
             \longrightarrow\lambda(2-\lambda)>0
\]
and
\[
 (A-(2-\lambda)\id)Dv_n=D(\lambda\id-A)v_n\longrightarrow0.
\]
Consequently $Dv_n/\norm{Dv_n}$ are unit approximate eigenvectors
at $2-\lambda$ for all sufficiently large $n$, proving the claim.
\end{proof}

\begin{proposition}[A dimension-independent restriction]
\label{supp:prop:discrete}
If four orthogonal projections on a nonzero complex Hilbert space satisfy
$\sum_{i=1}^4P_i=\alpha\id$ with $1<\alpha<2$, then
$\alpha=2-2/m$ for an integer $m\ge3$.
\end{proposition}
\begin{proof}
Put
\[
 A=P_1+P_2,\quad B=P_3+P_4=\alpha\id-A,\quad
 S=\spec(A),\quad \delta=4-2\alpha>0.
\]
Positivity and spectral mapping give
\begin{equation}\label{supp:eq:spectra}
 S\subseteq[0,\alpha],\qquad \spec(B)=\alpha-S.
\end{equation}
We first show
\begin{equation}\label{supp:eq:descent}
 \lambda\in S\setminus\{0,\delta/2\}\quad\Longrightarrow\quad
 \lambda-\delta\in S.
\end{equation}
Such a $\lambda$ lies in $(0,2)$, so
$r=2-\lambda\in S$ by Lemma~\ref{supp:lem:reflection}.
The equality $r=\alpha$ would give $\lambda=\delta/2$, which is excluded.
Hence $0<r<\alpha$ and
$\mu=\alpha-r\in\spec(B)\cap(0,2)$.
Reflecting the spectrum of $B$ gives $2-\mu\in\spec(B)$, and then
\[
 \alpha-(2-\mu)=\lambda-(4-2\alpha)=\lambda-\delta\in S.
\]
Starting at any $\lambda\in S$, apply~\eqref{supp:eq:descent} until reaching
$0$ or $\delta/2$. If neither were reached, $\lambda-j\delta$ would
belong to $S$ for every nonnegative integer $j$, contradicting
$S\subseteq[0,\infty)$.
The process therefore stops after finitely many steps, and
\begin{equation}\label{supp:eq:lattice}
 S\subseteq\{j\delta/2:j\in\mathbb N_0\}.
\end{equation}
There is a positive point of $S$. Otherwise the spectral theorem gives
$A=0$, so $B=\alpha\id$. Lemma~\ref{supp:lem:reflection} would force
$2-\alpha\in\spec(B)=\{\alpha\}$, contrary to $\alpha>1$.
Choose a positive $\lambda\in S$. Both $\lambda$ and $2-\lambda$ are
positive points of $S$, since $\lambda\le\alpha<2$.
By~\eqref{supp:eq:lattice}, their sum is $2=m\delta/2$ for a positive
integer $m$. Thus $\delta=4/m$ and $\alpha=2-2/m$.
Finally, $\alpha>1$ implies $m\ge3$.
\end{proof}

\begin{proposition}[All discrete parameters]\label{supp:prop:construction}
For every integer $m\ge3$ there are four real symmetric orthogonal
projections on $\mathbb C^m$ whose sum is $(2-2/m)\id$.
\end{proposition}
\begin{proof}
For $\lambda\in[0,2]$ set $s=\sqrt{\lambda(2-\lambda)}$. The matrices
\begin{equation}\label{supp:eq:blocks}
 R_\lambda=\frac12\begin{pmatrix}\lambda&s\\s&2-\lambda\end{pmatrix},
 \qquad T_\lambda=\frac12\begin{pmatrix}\lambda&-s\\-s&2-\lambda\end{pmatrix}
\end{equation}
are real symmetric projections, as direct multiplication using
$s^2=\lambda(2-\lambda)$ verifies. Their sum is
$\diag(\lambda,2-\lambda)$. A singleton $0$ or $1$ is a sum of
one-dimensional projections by $0=0+0$ or $1=1+0$.

Fix $m\ge3$ and put
\begin{equation}\label{supp:eq:constructAB}
 \alpha=2-\frac2m,\qquad
 A=\diag\left(\frac{2j}{m}:j=0,\ldots,m-1\right),\qquad B=\alpha\id-A.
\end{equation}
The entries of $B$ are $2(m-1-j)/m$, so its diagonal is the reverse
of that of $A$. In $A$ there is a singleton $0$. Each nonzero
entry with index $j$ pairs with the entry indexed by $m-j$, and their
sum is $2$. If $m$ is even, the fixed index $j=m/2$ contributes the
singleton $1$. Every other index belongs to a two-element pair.
Applying Eq.~\eqref{supp:eq:blocks} on these pairs, and permuting coordinates
back to their original order, gives $A=P_1+P_2$. Since $B$ has the
same multiset of diagonal entries, it also decomposes as $B=P_3+P_4$.
Their sum is $\alpha\id$ as required. The decompositions may use
different coordinate pairings.
\end{proof}

\section{A commuting representation for limits}
\label{supp:sec:limits}
\begin{lemma}\label{supp:lem:limit}
In any fixed finite bipartite scenario with binary outputs, every
$p$ in the closure of finite-dimensional tensor-product correlations
has a representation
\begin{equation}\label{supp:eq:commuting}
 p(a,b\mid i,j)=\ip{\psi}{A_i^aB_j^b\psi},
\end{equation}
where $\psi$ is a unit vector in a complex Hilbert space $\HH$,
$A_i,B_j$ are orthogonal projections, and $A_iB_j=B_jA_i$ for every
$i,j$. Here $A_i^1=A_i$, $A_i^0=\id-A_i$, $B_j^1=B_j$, and
$B_j^0=\id-B_j$; superscripts label outcomes, not powers.
\end{lemma}
\begin{proof}
Choose finite-dimensional tensor-product behaviors $p_n\to p$.
We first realize each $p_n$ by commuting projections and a vector state.
For a positive contraction $E$ on a finite-dimensional space $V$, define
on $V\oplus V$
\begin{equation}\label{supp:eq:dilation}
 \widehat E=
 \begin{pmatrix}E&\sqrt{E(\id-E)}\\
 \sqrt{E(\id-E)}&\id-E\end{pmatrix}.
\end{equation}
Functional calculus and direct multiplication give
$\widehat E^*=\widehat E=\widehat E^2$.
For the isometry $Jv=(v,0)$, $J^*\widehat E J=E$.
All binary effects for one party admit this dilation on the same doubled
space with the same $J$; no same-party commutation is required.
Dilate both local spaces and embed the density operator using the
tensor product of the two isometries. Compression preserves every
probability in~\eqref{supp:eq:tensor-model}.
Purify this density operator on a finite-dimensional ancilla, on which
all measurement operators act trivially.
This gives a Hilbert space $\HH_n$, a unit vector $\psi_n$, and
orthogonal projections $A_{i,n},B_{j,n}$ with cross-party commutation
that realize $p_n$.

Let $\mathcal W$ be the countable set of finite words in formal letters
$a_i,b_j$, including the empty word $1$. Write
$w_n=w(A_{1,n},\ldots,B_{1,n},\ldots)$.
For each pair $u,v\in\mathcal W$ the sequence
$\ip{u_n\psi_n}{v_n\psi_n}$ is bounded, since every word has norm
at most $1$. By a diagonal subsequence argument, assume that all
these sequences converge simultaneously.

Let $\mathcal V$ be the free complex vector space with basis
$\mathcal W$. Define, for finitely supported coefficients,
\begin{equation}\label{supp:eq:form}
 \left\langle\sum_u c_u u,\sum_v d_v v\right\rangle_0
 =\sum_{u,v}\overline{c_u}d_v
       \lim_n\ip{u_n\psi_n}{v_n\psi_n}.
\end{equation}
This is positive semidefinite: each quadratic form is a limit of
squared norms. Let $\mathcal N=\{\xi:\langle\xi,\xi\rangle_0=0\}$.
Cauchy--Schwarz for positive semidefinite forms shows that
$\mathcal N$ is a subspace orthogonal to all of $\mathcal V$.
Form the inner-product quotient $\mathcal V/\mathcal N$ and complete
it to $\HH$. The empty word gives $\psi=[1]$, with $\norm\psi=1$.

Left multiplication $L_i$ by $a_i$ is contractive for the seminorm:
if $\xi=\sum_wc_ww$, then
\begin{align*}
 \norm{L_i\xi}_0^2
 &=\lim_n\left\|A_{i,n}\sum_wc_ww_n\psi_n\right\|^2\\
 &\le\lim_n\left\|\sum_wc_ww_n\psi_n\right\|^2=\norm{\xi}_0^2.
\end{align*}
It therefore preserves $\mathcal N$ and induces a bounded operator
$A_i$ on $\HH$.
Finite-dimensional selfadjointness and idempotence imply
\[
 \langle L_i\xi,\eta\rangle_0=\langle\xi,L_i\eta\rangle_0,
 \qquad \norm{(L_i^2-L_i)\xi}_0=0
\]
for all $\xi,\eta\in\mathcal V$.
These identities extend by density, so $A_i^*=A_i=A_i^2$.
Left multiplication by $b_j$ similarly gives projections $B_j$.
The evaluation of $(a_ib_j-b_ja_i)\xi$ is zero for every $n$ and
$\xi\in\mathcal V$, whence $A_iB_j=B_jA_i$ first on a dense subspace
and then everywhere.

Finally, for every word $w$,
\[
 \ip{\psi}{w(A,B)\psi}=
 \lim_n\ip{\psi_n}{w(A_n,B_n)\psi_n}.
\]
Expanding the complements $\id-A_i$ and $\id-B_j$ proves
\eqref{supp:eq:commuting} for all output pairs.
\end{proof}

\begin{remark}
Only the inclusion of $\Cqa$ in the commuting-operator model is used.
No equality between these two correlation models is assumed.
The projection formulation of the commuting model also covers commuting
binary POVMs. If Alice's and Bob's positive contractions commute across
parties, apply Eq.~\eqref{supp:eq:dilation} with separate two-dimensional
ancillas on the same underlying Hilbert space. The matrix entries are
continuous functions of the original effects, so the dilated Alice and
Bob projections still commute. Compression to the common initial
ancilla vector preserves all joint probabilities.
\end{remark}

\section{Synchrony forces a scalar operator identity}
The following argument gives a commuting-model formulation of the
finite-dimensional moment reduction used in Ref.~\cite{MPS}.

A behavior is \emph{synchronous} if
$p(1,0\mid i,i)=p(0,1\mid i,i)=0$ for every $i$.
For $\alpha\in\mathbb R$, define
\begin{equation}\label{supp:eq:witness}
 F(p,\alpha)=\sum_{i,j=1}^4p(1,1\mid i,j)
 -2\alpha\sum_{i=1}^4p(1,1\mid i,i)+\alpha^2.
\end{equation}
\begin{lemma}\label{supp:lem:scalar}
If $p\in\Cqc$ is synchronous and $F(p,\alpha)=0$, then four
orthogonal projections on a nonzero Hilbert space have sum $\alpha\id$.
\end{lemma}
\begin{proof}
Take a commuting projection representation $A_i,B_j,\psi$ of $p$.
Cross-party commutation and the projection identities yield
\begin{align*}
 \norm{(A_i-B_i)\psi}^2
 &=\ip{\psi}{(A_i+B_i-2A_iB_i)\psi}\\
 &=p(1,0\mid i,i)+p(0,1\mid i,i)=0.
\end{align*}
Thus $A_i\psi=B_i\psi$.
Writing $S=\sum_iA_i$, we have
\[
 \ip{\psi}{A_iA_j\psi}=\ip{\psi}{A_iB_j\psi}=p(1,1\mid i,j),
 \qquad \ip{\psi}{A_i\psi}=p(1,1\mid i,i).
\]
It follows that
\begin{equation}\label{supp:eq:vector}
 \norm{(S-\alpha\id)\psi}^2=F(p,\alpha)=0.
\end{equation}
To promote this vector identity to an operator identity, consider
\[
 \mathcal D=\spanop\{w(B_1,\ldots,B_4)\psi:
                              w\text{ is a finite word}\},
 \qquad K=\overline{\mathcal D}.
\]
The empty word is included, so $\psi\in K$ and $K\ne\{0\}$.
For every Bob word $w$,
\[
 A_iw(B)\psi=w(B)A_i\psi=w(B)B_i\psi\in\mathcal D.
\]
Thus $A_i\mathcal D\subseteq\mathcal D$. For $x\in K$, choose
$x_n\in\mathcal D$ with $x_n\to x$. Boundedness of $A_i$ gives
$A_ix_n\to A_ix$, so closedness of $K$ implies $A_ix\in K$.
Hence the restrictions $P_i=A_i|_K$ are well-defined bounded operators
on $K$. They satisfy $P_i^2=P_i$, and for $x,y\in K$,
$\ip{P_ix}{y}=\ip{A_ix}{y}=\ip{x}{A_iy}=\ip{x}{P_iy}$.
Thus each $P_i$ is selfadjoint and is an orthogonal projection on $K$.

Equation~\eqref{supp:eq:vector} gives $(S-\alpha\id)\psi=0$.
Since every $A_i$ commutes with every $B_j$, $S$ commutes with
every Bob word. Consequently,
\[
 (S-\alpha\id)w(B)\psi=w(B)(S-\alpha\id)\psi=0.
\]
By linearity, $S-\alpha\id$ vanishes on $\mathcal D$. For any $x\in K$,
take $x_n\in\mathcal D$ converging to $x$. Since $\norm{S}\le4$,
\[
 \norm{(S-\alpha\id)x}
 =\norm{(S-\alpha\id)(x-x_n)}
 \le(4+|\alpha|)\norm{x-x_n}\longrightarrow0.
\]
Therefore $S|_K=\alpha\id_K$. Since $P_i=A_i|_K$, this is exactly
$\sum_iP_i=\alpha\id_K$.
\end{proof}

\section{The exact polynomial slice}
For $\alpha\in\mathbb R$ define $p_\alpha\in\mathbb R^{64}$ by
\begin{equation}\label{supp:eq:curve}
\begin{array}{c|cccc}
 &p_\alpha(1,1\mid i,j)&p_\alpha(1,0\mid i,j)
 &p_\alpha(0,1\mid i,j)&p_\alpha(0,0\mid i,j)\\[2pt]\hline
 i=j &\alpha/4&0&0&1-\alpha/4\\[3pt]
 i\ne j&\dfrac{\alpha(\alpha-1)}{12}
 &\dfrac{\alpha(4-\alpha)}{12}
 &\dfrac{\alpha(4-\alpha)}{12}
 &\dfrac{(3-\alpha)(4-\alpha)}{12}
\end{array}
\end{equation}
Every coordinate is a polynomial with rational coefficients. For
$1<\alpha<2$ the entries are nonnegative and normalized, and the
outcome-$1$ marginals are $\alpha/4$, independent of the opposite
input. Thus the curve consists of synchronous nonsignaling behaviors
on this interval.

\begin{theorem}[Exact slice]\label{supp:thm:exact}
For $t\in\{q,qa,qc\}$,
\begin{equation}\label{supp:eq:exact}
 \{\alpha\in(1,2):p_\alpha\in\Ct\}
 =\{2-2/m:m=3,4,\ldots\}.
\end{equation}
\end{theorem}
\begin{proof}
For necessity it suffices to work in $\Cqc$. Equation~\eqref{supp:eq:curve}
gives
\[
 \sum_i p_\alpha(1,1\mid i,i)=\alpha,\qquad
 \sum_{i,j}p_\alpha(1,1\mid i,j)
 =\alpha+\alpha(\alpha-1)=\alpha^2.
\]
Hence the synchronous behavior $p_\alpha$ satisfies
$F(p_\alpha,\alpha)=0$, where $F$ is defined in Eq.~\eqref{supp:eq:witness}.
Lemma~\ref{supp:lem:scalar} and Proposition~\ref{supp:prop:discrete} force
$\alpha=2-2/m$ for an integer $m\ge3$.

For sufficiency, fix $m\ge3$, put $\alpha=2-2/m$, and choose
$P_1,\ldots,P_4$ on $\mathbb C^m$ using
Proposition~\ref{supp:prop:construction}. On dimension $d=24m$ let
\[
 Q_i=\bigoplus_{\pi\in S_4}P_{\pi(i)},\qquad \tau(X)=\Tr(X)/d.
\]
These are projections and $\sum_iQ_i=\alpha\id$.
For each fixed $i$, every label occurs $6$ times as $\pi(i)$, so
\begin{equation}\label{supp:eq:marginal}
 \tau(Q_i)=\frac{6}{24m}\sum_{r=1}^4\Tr(P_r)=\frac\alpha4.
\end{equation}
For distinct $i,j$, every ordered pair of distinct labels occurs $2$
times as $(\pi(i),\pi(j))$. Expanding the square of the scalar sum,
\[
 \sum_{r\ne s}\Tr(P_rP_s)
 =\Tr\left[\left(\sum_rP_r\right)^2\right]-\sum_r\Tr(P_r)
 =m(\alpha^2-\alpha),
\]
gives
\begin{equation}\label{supp:eq:pair}
 \tau(Q_iQ_j)=\frac{2m(\alpha^2-\alpha)}{24m}
 =\frac{\alpha(\alpha-1)}{12}\qquad (i\ne j).
\end{equation}
Use the maximally entangled unit vector
$\Phi_d=d^{-1/2}\sum_{r=1}^d e_r\otimes e_r$.
For arbitrary matrices $X,Y\in M_d(\mathbb C)$,
\begin{equation}\label{supp:eq:ME}
 \ip{\Phi_d}{(X\otimes Y^{\mathsf T})\Phi_d}
 =\frac1d\sum_{r,s}X_{rs}Y_{sr}=\tau(XY).
\end{equation}
The transpose is taken in the displayed basis. Alice measures
$Q_i$ and $\id-Q_i$; Bob measures $Q_j^{\mathsf T}$ and
$\id-Q_j^{\mathsf T}$. Equations~\eqref{supp:eq:marginal}--\eqref{supp:eq:ME}
give the $(1,1)$ probabilities and both marginals of $p_\alpha$.
For $i=j$, mixed outcomes have probability zero because
$Q_i(\id-Q_i)=0$. For $i\ne j$ the other probabilities are
\begin{align*}
 p(1,0\mid i,j)=p(0,1\mid i,j)
 &=\frac\alpha4-\frac{\alpha(\alpha-1)}{12}
 =\frac{\alpha(4-\alpha)}{12},\\
 p(0,0\mid i,j)
 &=1-\frac\alpha2+\frac{\alpha(\alpha-1)}{12}
 =\frac{(3-\alpha)(4-\alpha)}{12}.
\end{align*}
Thus $p_{2-2/m}\in\Cq$. The inclusions of the three models finish
the proof.
\end{proof}

\section{Semialgebraic and semianalytic obstructions}
\begin{corollary}\label{supp:cor:nonsemi}
Every set $C$ with $\Cq\subseteq C\subseteq\Cqc$ is nonsemialgebraic.
In particular this holds for $\Cq$, $\Cqa$, and $\Cqc$.
\end{corollary}
\begin{proof}
The two inclusions and Theorem~\ref{supp:thm:exact} force the same exact
slice for $C$. If $C$ were semialgebraic, substituting the polynomial
coordinates of $p_\alpha$ into its finite Boolean description, and
adjoining $1<\alpha<2$, would give a semialgebraic description of
$\{2-2/m:m\ge3\}$. A finite collection of nonzero univariate
polynomials has finitely many roots. On each complementary interval
all signs, and hence Boolean truth values, are constant. Identically
zero polynomials contribute constant truth values. Every such set
is therefore a finite union of intervals and individual points,
which cannot be infinite and countable.
\end{proof}

\begin{proposition}[Failure of semianalyticity at a classical point]
\label{supp:prop:analytic}
The behavior $p_2$ is classical local. Every set
$\Cq\subseteq C\subseteq\Cqc$ fails to be semianalytic at $p_2$:
there is no open neighborhood $U$ of $p_2$ in $\mathbb R^{64}$ and
no finite collection of real analytic functions on $U$ whose signs,
combined by a finite Boolean formula, describe $C\cap U$.
\end{proposition}
\begin{proof}
Let the parties share a random variable uniformly distributed among
the six two-element subsets $S\subset\{1,2,3,4\}$. On input $i$,
both parties output $1$ if and only if $i\in S$.
Each label lies in three of the six subsets, and each pair of distinct
labels lies in one. Therefore the marginals are $1/2$, distinct-input
$(1,1)$ probabilities are $1/6$, and equal inputs always give equal
outputs. This is exactly $p_2$. Being a finite mixture of deterministic
local strategies, it belongs to $\Cq$.

Suppose a finite analytic description using functions $f_1,\ldots,f_N$
on $U$ existed. Continuity of $p_\alpha$ gives an open interval $J$
containing $2$ such that $p_\alpha\in U$ for $\alpha\in J$.
Each composition $g_r(\alpha)=f_r(p_\alpha)$ is real analytic on $J$.
If its germ at $2$ is zero, it vanishes on some neighborhood of $2$.
Otherwise its zeros are isolated, so there is a punctured neighborhood
of $2$ where it does not vanish. Taking the minimum of finitely many
neighborhood sizes, all $g_r$ have constant signs on an interval
$(2-\epsilon,2)\subset J\cap(1,2)$. Thus the Boolean formula has a
constant truth value on that interval. Since it contains infinitely
many $2-2/m$, this value must be true. The interval also contains
parameters not of that form, contradicting Theorem~\ref{supp:thm:exact}.
\end{proof}

\begin{remark}[Scope of the analytic statement]
Analyticity on a neighborhood containing the accumulation point is
essential. On the punctured domain $\alpha<2$, the analytic equation
$\sin(2\pi/(2-\alpha))=0$ can describe the discrete parameter sequence
when restricted to $\alpha\in(1,2)$. Nor does the proposition rule out
general projections of analytic sets: the conditions
$(2-\alpha)t=2$, $\sin(\pi t)=0$, and $t\ge3$ give an analytic lift
of that sequence using an unbounded auxiliary variable.
The finite-lift exclusion below is specifically semialgebraic.
The local analytic statement addresses the analytic-description
question in Ref.~\cite[Problem~2.10]{Tsirelson} under the explicit
requirement that the defining functions be analytic at the behavior
being described. No equivalence of tensor and commuting models from
that historical formulation is assumed here.
\end{remark}

\section{Finite lifts and outer approximations}
\begin{corollary}\label{supp:cor:lifts}
If $\Cq\subseteq C\subseteq\Cqc$, then $C$ is not the projection
of a semialgebraic set in any finite-dimensional real space.
In particular, $C$ has no exact finite semidefinite representation,
with arbitrary real coefficients and finitely many auxiliary variables.
\end{corollary}
\begin{proof}
The Tarski--Seidenberg theorem says that coordinate projections of
semialgebraic sets are semialgebraic. An alleged finite lift would
therefore contradict Corollary~\ref{supp:cor:nonsemi}. For semidefinite
lifts, a real symmetric affine matrix inequality is equivalent to
nonnegativity of finitely many principal minors. A complex Hermitian
constraint can instead be represented by a real symmetric block
constraint on real and imaginary parts. Finite collections of these
constraints, together with affine equalities, define semialgebraic
sets. Their projections are therefore excluded.
\end{proof}

\begin{proposition}[A tail interval is unavoidable]\label{supp:prop:tail}
Let $R\subset\mathbb R^{64}$ be semialgebraic with $\Cq\subseteq R$.
There is $\beta\in(1,2)$ such that $p_\alpha\in R$ for every
$\alpha\in(\beta,2)$.
\end{proposition}
\begin{proof}
The inverse image $T_R=\{\alpha\in(1,2):p_\alpha\in R\}$ is
semialgebraic and contains $2-2/m$ for every $m\ge3$.
In a finite polynomial description of $T_R$, choose $\beta<2$
larger than every root strictly below $2$ sufficiently close to $2$,
and larger than $1$. On $(\beta,2)$ every polynomial sign is constant.
The interval contains members of the sequence, so its constant
membership value must be true.
\end{proof}

For the standard complete NPA hierarchy~\cite{NPA}, a fixed finite
level $\mathcal Q_\ell$ is the projection of a finite semidefinite
feasibility set. It contains $\Cqc$ and is semialgebraic. Hence
Proposition~\ref{supp:prop:tail} provides $\beta_\ell\in(1,2)$ with
\[
 p_\alpha\in\mathcal Q_\ell\quad\text{for all }\alpha\in(\beta_\ell,2).
\]
For every parameter in this interval except the discrete values
$2-2/m$, Theorem~\ref{supp:thm:exact} gives $p_\alpha\notin\Cqc$.
Thus each fixed finite level strictly contains $\Cqc$ and cannot
equal $\Cqa$ either. The completeness identity
$\bigcap_\ell\mathcal Q_\ell=\Cqc$ means that each fixed
$p_\alpha\notin\Cqc$ is nevertheless excluded at some finite level.
These statements do not assert strict inclusion between all successive
levels, prohibit finite exactness for a particular Bell functional,
or supply a quantitative convergence rate or robustness bound.

\section{Physical consequences and comparison with smaller scenarios}
\label{supp:sec:physical}

\begin{proposition}[Nonlocal behavior with local correlators]
\label{supp:prop:local-correlators}
No $p_\alpha$ with $1<\alpha<2$ is Bell local. Nevertheless its matrix
of binary correlators admits a local model for every such $\alpha$.
\end{proposition}
\begin{proof}
A local behavior is a convex combination of deterministic response
functions. Since $p_\alpha$ is synchronous, nonnegativity of the
weights forces each response pair appearing with positive weight
to agree on every equal input. Write the resulting shared bits as
$z_i\in\{0,1\}$, and put $N=\sum_i z_i$. The marginals and second
moments of $p_\alpha$ give $\mathbb E N=\alpha$ and
$\mathbb E N^2=\alpha^2$. Thus $\mathbb E(N-\alpha)^2=0$.
Since $N$ is integer-valued, this is impossible for $1<\alpha<2$.
In particular, every realizable point of the sequence in
Theorem~\ref{supp:thm:exact} is nonlocal.

For the second claim, direct substitution in
Eq.~\eqref{supp:eq:curve} gives diagonal correlators $1$ and
constant off-diagonal correlator
$c_\alpha=((\alpha-2)^2-1)/3\in(-1/3,0)$.
Let $\mu_{\mathrm{bal}}$ be the uniform measure on the six strings
$s\in\{-1,1\}^4$ with two plus signs, and $\mu_{\mathrm{ind}}$
the uniform measure on all sixteen strings. In both models the parties
share $s$ and output $s_i$ and $s_j$.
Every single-party mean is zero and every diagonal correlator is one.
For distinct indices the correlator is $-1/3$ under
$\mu_{\mathrm{bal}}$ and $0$ under $\mu_{\mathrm{ind}}$.
The convex mixture
\[
 (-3c_\alpha)\mu_{\mathrm{bal}}+(1+3c_\alpha)\mu_{\mathrm{ind}}
\]
therefore has exactly the desired correlators. Its probabilities differ
from $p_\alpha$ because its single-party means vanish, whereas those
of $p_\alpha$ equal $1-\alpha/2$. The correlator projection thus loses
the compatibility information responsible for exact membership.
\end{proof}

\begin{corollary}[Exact dimension growth at a classical limit]
\label{supp:cor:dimension}
Let $m\ge3$, $\alpha_m=2-2/m$, and $q_m=m/\gcd(m,2)$.
Every finite-dimensional realization of $p_{\alpha_m}$ requires
\[
 d_A\ge q_m,\qquad d_B\ge q_m.
\]
The same lower bounds hold on the quantum dimension allowed to each
party when arbitrary shared classical randomness is free.
If $D_m$ is the least $D$ permitting a realization with both local
dimensions at most $D$, and $D_m^{\mathrm{sr}}$ is the corresponding
quantity with free shared randomness, then
\begin{equation}\label{supp:eq:dimension-bounds}
 q_m\le D_m^{\mathrm{sr}}\le m,
 \qquad q_m\le D_m\le24m.
\end{equation}
In particular, both quantities are
$\Theta((2-\alpha_m)^{-1})$ along the sequence, although
$p_{\alpha_m}\to p_2$ and $p_2$ is classical local.
\end{corollary}
\begin{proof}
We use the moment-to-scalar mechanism of Ref.~\cite{MPS} to prove a
rank constraint for any synchronous behavior $r$
with $F(r,\alpha)=0$, where $F$ is defined in
Eq.~\eqref{supp:eq:witness}.
Consider first a pure-state realization $\psi$ by binary POVM effects
$E_i$ and $G_i$, and put $X_i=E_i\otimes\id$ and
$Y_i=\id\otimes G_i$.
Synchrony and cross-party commutation give
\begin{align*}
 0&=r(1,0\mid i,i)+r(0,1\mid i,i)\\
  &=\norm{(X_i-Y_i)\psi}^2
    +\ip{\psi}{(X_i-X_i^2+Y_i-Y_i^2)\psi}.
\end{align*}
Each term is nonnegative because $E_i$ and $G_i$ are positive
contractions. Hence $X_i\psi=Y_i\psi$ and
$(X_i-X_i^2)\psi=(Y_i-Y_i^2)\psi=0$.

Write the Schmidt decomposition
$\psi=\sum_{k=1}^{s}\sqrt{\lambda_k}\,u_k\otimes v_k$,
with every $\lambda_k>0$, and let $V_A=\operatorname{span}\{u_k\}$
and $V_B=\operatorname{span}\{v_k\}$.
The relation $X_i\psi=Y_i\psi$ makes $V_A$ invariant under $E_i$:
projecting it onto $V_A^\perp\otimes\mathcal H_B$ gives
$\sum_k\sqrt{\lambda_k}(\Pi_{V_A^\perp}E_i u_k)\otimes v_k=0$,
so every coefficient vanishes. The analogous argument gives
invariance of $V_B$ under $G_i$.
Thus the restrictions $P_i=E_i|_{V_A}$ and $Q_i=G_i|_{V_B}$
are selfadjoint. The preceding identities for $X_i-X_i^2$ and
$Y_i-Y_i^2$, together with the full Schmidt rank on these supports,
show that $P_i^2=P_i$ and $Q_i^2=Q_i$.
In particular no increase in local dimension is required to obtain
projections on the Schmidt supports.

For $S_A=\sum_iP_i$, the same moment calculation as in
Lemma~\ref{supp:lem:scalar} now gives
\[
 F(r,\alpha)=\norm{(S_A\otimes\id-\alpha\id)\psi}^2\ge0.
\]
If this expression is zero, full Schmidt rank implies
$S_A=\alpha\id_{V_A}$. Summing $X_i\psi=Y_i\psi$ gives likewise
$\sum_iQ_i=\alpha\id_{V_B}$.
Taking the ordinary matrix trace on either $s$-dimensional support
therefore yields
\begin{equation}\label{supp:eq:rank-integrality}
 \alpha s=\sum_i\operatorname{rank}P_i\in\mathbb Z.
\end{equation}
For $\alpha=\alpha_m$, its denominator in lowest terms is $q_m$;
thus $q_m$ divides $s$, and both local dimensions are at least $q_m$.

These conclusions extend to mixed states and shared randomness.
In a finite-dimensional mixed-state strategy, decompose the state
into pure states of positive weight. Since all mismatch probabilities
are nonnegative, synchrony of the mixture forces synchrony of each
pure component. The preceding argument shows $F(r,\alpha)\ge0$
for each such component. For fixed $\alpha$, $F$ is affine in $r$;
hence $F=0$ for the mixture forces $F=0$ in every positive-weight
component. Each component then obeys
Eq.~\eqref{supp:eq:rank-integrality} within the original local spaces.
Consequently the same statements hold for every mixed-state strategy.

For arbitrary shared randomness, write the observed behavior as
$r=\int r_\lambda\,d\mu(\lambda)$, where each conditional strategy
has the prescribed local dimension bounds. There are finitely many
mismatch probabilities. Their nonnegativity and zero integrals imply
that $r_\lambda$ is synchronous for almost every $\lambda$.
For these branches $F(r_\lambda,\alpha)\ge0$, and affinity gives
$\int F(r_\lambda,\alpha)\,d\mu(\lambda)=F(r,\alpha)=0$.
Thus $F(r_\lambda,\alpha)=0$ almost surely, so almost every branch
requires both local dimensions to be at least $q_m$.
The behavior $p_{\alpha_m}$ has synchrony and $F=0$ by
Theorem~\ref{supp:thm:exact}, establishing both lower bounds.

The construction in Theorem~\ref{supp:thm:exact} uses local dimension
$24m$, proving the upper bound for $D_m$.
If shared randomness is free, implement the same symmetrization by
choosing $\pi\in S_4$ uniformly and measuring
$P_{\pi(i)}$ and $P_{\pi(j)}^{\mathsf T}$ on $\Phi_m$.
Averaging these dimension-$m$ strategies produces exactly the same
probabilities as the block-direct-sum construction. Hence
$D_m^{\mathrm{sr}}\le m$.
Finally $m/2\le q_m\le m$ and $m=2/(2-\alpha_m)$ give the claimed
asymptotic bounds. They concern exact realization; no lower bound at
a fixed nonzero noise tolerance is asserted.
\end{proof}

\begin{corollary}[Threshold within the synchronous binary sector]
\label{supp:cor:synchronous-threshold}
Let $C_t^{\mathrm{sync}}(n,2)$ denote the synchronous sector with
$n$ inputs per party and two outputs, for $t\in\{q,qa,qc\}$.
These sets are semialgebraic for $n\le3$ and nonsemialgebraic for
$n\ge4$.
\end{corollary}
\begin{proof}
Russell gives a uniform finite-dimensional realization of every
point in $C_q^{\mathrm{sync}}(3,2)$~\cite{RussellGeom}, and establishes
$C_q^{\mathrm{sync}}(3,2)=C_{qc}^{\mathrm{sync}}(3,2)$~\cite{RussellConnes}.
The intermediate closed model is the same set. Fixed-dimensional
matrix parametrization, with the linear synchrony conditions,
therefore proves semialgebraicity. Restricting inputs gives the
$n<3$ cases; conversely extra inputs can be duplicated, so these
restrictions are onto.
For $n=4$, every $p_\alpha$ in Theorem~\ref{supp:thm:exact} is
synchronous. Its discrete inverse image thus rules out semialgebraicity
of the synchronous sector itself. Duplicating inputs embeds this
example into any $n>4$, as formalized below. This threshold concerns
synchronous behaviors with equal numbers of binary inputs on the
two sides. It does not establish the smallest nonsynchronous Bell
scenario with a nonsemialgebraic quantum set.
\end{proof}

\begin{corollary}[Inheritance by larger scenarios]
\label{supp:cor:larger-scenarios}
For every fixed $n_A,n_B\ge4$ and $m_A,m_B\ge2$, each of the
three corresponding quantum correlation sets is nonsemialgebraic.
\end{corollary}
\begin{proof}
Define a linear embedding $E$ of four-input binary behaviors by
repeating an existing measurement for each added input and assigning
zero probability to added outcomes. It preserves realizability in
all three models. In the reverse direction, restrict to the original
inputs and merge added outcomes with outcome $0$. This linear map
$R$ preserves realizability, is continuous, and satisfies $REp=p$.
Continuity handles the closed finite-dimensional model, while
repeating or merging measurements handles the other two models.
Thus $Ep$ belongs to the larger quantum set if and only if $p$
belongs to its four-input binary counterpart. The inverse image
under $E$ of a semialgebraic set would be semialgebraic, contradicting
Corollary~\ref{supp:cor:nonsemi}. Duplicating inputs also preserves
synchrony, giving the final step of
Corollary~\ref{supp:cor:synchronous-threshold}.
\end{proof}


\begin{thebibliography}{99}

\bibitem{Bell}
J.~S.~Bell, On the Einstein Podolsky Rosen paradox,
\href{https://doi.org/10.1103/PhysicsPhysiqueFizika.1.195}{Physics Physique Fizika \textbf{1}, 195 (1964)}.

\bibitem{Brunner}
N.~Brunner, D.~Cavalcanti, S.~Pironio, V.~Scarani, and S.~Wehner,
Bell nonlocality,
\href{https://doi.org/10.1103/RevModPhys.86.419}{Rev. Mod. Phys. \textbf{86}, 419 (2014)}.

\bibitem{Goh}
K.~T.~Goh, J.~Kaniewski, E.~Wolfe, T.~V\'ertesi, X.~Wu, Y.~Cai,
Y.-C.~Liang, and V.~Scarani, Geometry of the set of quantum correlations,
\href{https://doi.org/10.1103/PhysRevA.97.022104}{Phys. Rev. A \textbf{97}, 022104 (2018)}.

\bibitem{NPA2007}
M.~Navascu\'es, S.~Pironio, and A.~Ac\'in,
Bounding the set of quantum correlations,
\href{https://doi.org/10.1103/PhysRevLett.98.010401}{Phys. Rev. Lett. \textbf{98}, 010401 (2007)}.

\bibitem{NPA}
M.~Navascu\'es, S.~Pironio, and A.~Ac\'in,
A convergent hierarchy of semidefinite programs characterizing the set of quantum correlations,
\href{https://doi.org/10.1088/1367-2630/10/7/073013}{New J. Phys. \textbf{10}, 073013 (2008)}.

\bibitem{Tsirelson87}
B.~S.~Tsirel'son, Quantum analogues of the Bell inequalities.
The case of two spatially separated domains,
\href{https://doi.org/10.1007/BF01663472}{J. Sov. Math. \textbf{36}, 557 (1987)}.

\bibitem{Masanes}
Ll.~Masanes, Extremal quantum correlations for $N$ parties with two dichotomic observables per site,
\href{https://arxiv.org/abs/quant-ph/0512100}{arXiv:quant-ph/0512100}.

\bibitem{Le}
T.~P.~Le, C.~Meroni, B.~Sturmfels, R.~F.~Werner, and T.~Ziegler,
Quantum correlations in the minimal scenario,
\href{https://doi.org/10.22331/q-2023-03-16-947}{Quantum \textbf{7}, 947 (2023)}.

\bibitem{BCR}
J.~Bochnak, M.~Coste, and M.-F.~Roy,
\href{https://doi.org/10.1007/978-3-662-03718-8}{\emph{Real Algebraic Geometry}}
(Springer, Berlin, 1998).

\bibitem{Tsirelson}
B.~S.~Tsirelson, Some results and problems on quantum Bell-type inequalities,
\href{https://ma.huji.ac.il/~ohadfeld/Tsirelson/download/hadron.pdf}{Hadronic J. Suppl. \textbf{8}, 329 (1993)}, Problem~2.10.

\bibitem{Ozawa2013}
N.~Ozawa, About the Connes embedding conjecture---algebraic approaches,
\href{https://doi.org/10.1007/s11537-013-1280-5}{Jpn. J. Math. \textbf{8}, 147 (2013)}, Sec.~20.

\bibitem{Slofstra}
W.~Slofstra, The set of quantum correlations is not closed,
\href{https://doi.org/10.1017/fmp.2018.3}{Forum Math. Pi \textbf{7}, e1 (2019)}.

\bibitem{DPP}
K.~Dykema, V.~I.~Paulsen, and J.~Prakash,
Non-closure of the set of quantum correlations via graphs,
\href{https://doi.org/10.1007/s00220-019-03301-1}{Commun. Math. Phys. \textbf{365}, 1125 (2019)}.

\bibitem{MIPRE}
Z.~Ji, A.~Natarajan, T.~Vidick, J.~Wright, and H.~Yuen,
$\mathrm{MIP}^{*}=\mathrm{RE}$,
\href{https://arxiv.org/abs/2001.04383}{arXiv:2001.04383}.

\bibitem{FMS}
H.~Fu, C.~A.~Miller, and W.~Slofstra,
The membership problem for constant-sized quantum correlations is undecidable,
\href{https://doi.org/10.1007/s00220-024-05229-7}{Commun. Math. Phys. \textbf{406}, 96 (2025)}.

\bibitem{PNA}
S.~Pironio, M.~Navascu\'es, and A.~Ac\'in,
Convergent relaxations of polynomial optimization problems with noncommuting variables,
\href{https://doi.org/10.1137/090760155}{SIAM J. Optim. \textbf{20}, 2157 (2010)}.

\bibitem{PSSTW}
V.~I.~Paulsen, S.~Severini, D.~Stahlke, I.~G.~Todorov, and A.~Winter,
Estimating quantum chromatic numbers,
\href{https://doi.org/10.1016/j.jfa.2016.01.010}{J. Funct. Anal. \textbf{270}, 2188 (2016)}.

\bibitem{KRS}
S.~A.~Kruglyak, V.~I.~Rabanovich, and Yu.~S.~Samoilenko,
On sums of projections,
\href{https://doi.org/10.1023/A:1020193804109}{Funct. Anal. Appl. \textbf{36}, 182 (2002)}.

\bibitem{MPS}
L.~Man\v{c}inska, J.~Prakash, and C.~Schafhauser,
Constant-sized robust self-tests for states and measurements of unbounded dimension,
\href{https://doi.org/10.1007/s00220-024-05122-3}{Commun. Math. Phys. \textbf{405}, 221 (2024)}.

\bibitem{SM}
See the \hyperref[supplement-start]{Supplemental Material appended below}
for full proofs, dimension bounds, and the comparison with three-input
synchronous correlations.

\bibitem{NV}
M.~Navascu\'es and T.~V\'ertesi, Bounding the set of finite dimensional quantum correlations,
\href{https://doi.org/10.1103/PhysRevLett.115.020501}{Phys. Rev. Lett. \textbf{115}, 020501 (2015)}.

\bibitem{RussellGeom}
T.~B.~Russell, Geometry of the set of synchronous quantum correlations,
\href{https://doi.org/10.1063/1.5115010}{J. Math. Phys. \textbf{61}, 052201 (2020)}.

\bibitem{RussellConnes}
T.~B.~Russell, Two-outcome synchronous correlation sets and Connes' embedding problem,
\href{https://doi.org/10.26421/QIC20.5-6-1}{Quantum Inf. Comput. \textbf{20}, 361 (2020)}.

\bibitem{AlmostQuantum}
M.~Navascu\'es, Y.~Guryanova, M.~J.~Hoban, and A.~Ac\'in,
Almost quantum correlations,
\href{https://doi.org/10.1038/ncomms7288}{Nat. Commun. \textbf{6}, 6288 (2015)}.

\end{thebibliography}
\end{document}